\documentclass[11pt]{article}

\usepackage{iftex}
\ifPDFTeX
  \usepackage[utf8]{inputenc}
  \usepackage[T1]{fontenc}
\fi
\usepackage{amsmath,amssymb,amsthm}
\usepackage{mathtools}
\usepackage{bm}
\usepackage[margin=1in]{geometry}
\usepackage{booktabs}
\usepackage{enumitem}
\usepackage{hyperref}
\usepackage{url}

\theoremstyle{plain}
\newtheorem{theorem}{Theorem}
\newtheorem{lemma}[theorem]{Lemma}
\newtheorem{corollary}[theorem]{Corollary}
\theoremstyle{definition}
\newtheorem{assumption}{Assumption}
\theoremstyle{remark}
\newtheorem*{remark}{Remark}

\newcommand{\R}{\mathbb{R}}
\newcommand{\Lcone}{\mathcal{L}}
\newcommand{\op}{\mathrm{op}}
\newcommand{\loc}{\mathrm{loc}}
\newcommand{\eps}{\varepsilon}
\newcommand{\norm}[1]{\left\lVert #1 \right\rVert}
\newcommand{\Ot}{\widetilde{O}}              
\newcommand{\xopt}{x^{*}}
\DeclareMathOperator*{\argmin}{arg\,min}
\newcommand{\trans}{^{\!\top}}

\title{Warm-Start Interior-Point Methods for\\ Online Second-Order Cone Programming}
\author{Krishna Harish}
\date{June 4, 2026}

\begin{document}
\maketitle

\begin{abstract}
We analyze the computational complexity of solving a sequence of related second-order
cone programs (SOCPs) whose right-hand-side data $b_t$ varies between rounds. The standard
primal--dual interior-point algorithm solves each round at cost $\Ot(n^{2.5}\log(1/\eps))$ from a
cold start. We show that when the per-round perturbation $\norm{b_t-b_{t-1}}_2$ is bounded by a
problem-specific threshold $\delta$, Newton's method warm-started at the previous round's solution
$\xopt_{t-1}$ converges to $\xopt_t$ to accuracy $\eps$ in $O(\log\log(1/\eps))$ iterations. Over $T$
rounds the total cost is $\Ot\!\left(n^{2.5}\log(1/\eps)+T n^{2}\log\log(1/\eps)\right)$, compared to
$\Ot(T n^{2.5}\log(1/\eps))$ for cold start at each round; the per-round speedup for large $T$ is
$\Theta(\sqrt{n}\,\log(1/\eps)/\log\log(1/\eps))$. The argument combines an infinitesimal local-norm
sensitivity bound on the central-path optimum (Lemma~\ref{lem:infsens}), a self-concordant
finite-difference corollary (Corollary~\ref{cor:findiff}), and the standard quadratic-convergence
basin of Newton's method on a self-concordant barrier. The local-norm formulation circumvents the
rank-deficiency issues of Euclidean sensitivity bounds for fat constraint matrices. The Lipschitz
constant $L_{\loc}(\eta,b)$ depends on the central-path parameter $\eta$ through $w(\xopt(\eta,b),b)$
and is finite at every fixed $\eta$, which is what the warm-start argument needs. A multi-seed
experiment on bounded SOCPs with $n=50$, $p=100$ confirms a 30--70$\times$ per-round speedup across
the predicted regime.
\end{abstract}

\section{Introduction}

Online second-order cone programming arises when the data of a parametric SOCP arrive sequentially
in time. At each round $t=1,2,\dots,T$ a solver receives a perturbation of the previous round's data
and is required to produce an optimal solution before the next round's data arrive. Applications
include model predictive control of constrained linear systems, online portfolio rebalancing under
chance constraints, sliding-window robust optimization, and real-time scheduling under conic side
constraints; see~\cite{lobo} for applications of second-order cone programming. In all of these the round-to-round perturbation is small relative to the problem size.

The cost of solving each round from a cold start is the standard Nesterov--Todd~\cite{nt} primal--dual
interior-point complexity: $O(\sqrt{m}\,\log(1/\eps))$ Newton iterations on a sequence of
central-path subproblems, each iteration costing $O(n^{2}+mp)$ arithmetic operations for an SOCP with
$m$ Lorentz constraints of dimension $p+1$. For typical parameter regimes $p=O(1)$ and $m=O(n)$ the
per-round cost is $\Ot(n^{2.5}\log(1/\eps))$ and the total cost over $T$ rounds is
$\Ot(T n^{2.5}\log(1/\eps))$. This is wasteful: it ignores the fact that consecutive rounds have
nearly identical solutions.

We prove that the cold-start factor $\sqrt{n}\,\log(1/\eps)$ in the per-round iteration count can be
replaced by $\log\log(1/\eps)$ when the per-round perturbation is small enough relative to a
problem-specific threshold. The argument has three ingredients. First, an infinitesimal sensitivity
bound (Lemma~\ref{lem:infsens}) that controls the directional derivative of $\xopt(b)$ in the local
norm at the central-path point, with an explicit Lipschitz constant depending on the central-path
point $\xopt(\eta,b)$ through $w$, $\rho_A$, and $\norm{u}_2$. Second, a self-concordant
finite-difference corollary (Corollary~\ref{cor:findiff}) that integrates the infinitesimal bound to
a Lipschitz bound on $\norm{\xopt(b')-\xopt(b)}$ at the cost of an explicit constant factor. Third,
the classical quadratic-convergence basin of Newton's method on a self-concordant barrier
(Lemma~\ref{lem:quadconv}), which gives $O(\log\log(1/\eps))$ iterations once the starting point is
in the basin. Composing the three ingredients gives the main theorem.

\subsection{Contributions}

The technical core of the paper is a sensitivity bound on the central-path optimum $\xopt(\eta,b)$ in
the local norm at the central-path point. Lemma~\ref{lem:infsens} gives the infinitesimal version:
for an SOCP of the form $\min c\trans x$ subject to $\norm{Ax+b}_2\le c_0\trans x + d$, under strict
feasibility and a mild non-degeneracy condition (Assumption~\ref{ass:nondeg}),
\[
  L_{\loc}(\eta,b)\;\le\;\frac{2}{w(\xopt(\eta,b),b)}\Big(\rho_A(\eta,b)+\sqrt{2}\,\norm{u(\xopt(\eta,b),b)}_2\Big),
\]
where $w(\xopt,b)=s(\xopt)^2-\norm{u}_2^2$ is the cone-distance quantity at the central-path point,
$\rho_A(\eta,b):=\norm{A H^{-1}A\trans}_{\op}^{1/2}$ is a problem-specific conditioning quantity, and
$u:=A\xopt+b$. Corollary~\ref{cor:findiff} integrates this to a finite-difference Lipschitz bound
$\norm{\xopt(\eta,b')-\xopt(\eta,b)}_{\xopt(\eta,b')}\le 2 L_{\loc}\norm{b'-b}_2$ for $b'$ in a
self-concordant neighborhood of $b$, with the factor $2$ absorbing the change of the local norm along
the path. The local-norm formulation circumvents the rank-deficiency issues that arise in Euclidean
sensitivity bounds when $A$ is not square.

Combining these with the standard quadratic-convergence basin of Newton's method on a self-concordant
barrier gives the main result. Theorem~\ref{thm:warmconv} shows that under the per-round perturbation
bound $\norm{b_t-b_{t-1}}_2\le 1/(20 L_{\loc})$ at the working parameter $\eta=2/\eps$, Newton's
method on the round-$t$ log-barrier started at $\xopt_{t-1}$ converges to within $\eps$ of $\xopt_t$ in
$O(\log\log(1/\eps))$ iterations. Theorem~\ref{thm:cost} packages this into a total amortized cost of
$\Ot(n^{2.5}\log(1/\eps)+T n^{2}\log\log(1/\eps))$ over $T$ rounds, compared with
$\Ot(T n^{2.5}\log(1/\eps))$ for cold start at every round. The analysis extends to multiple Lorentz
constraints (Section~\ref{sec:multi}).

A multi-seed numerical experiment in Section~\ref{sec:numerics} verifies the predicted scaling. With
$n=50$, $p=100$, cold-start IPM takes a consistent $\approx 35$ Newton iterations per round;
warm-start takes 0--3 iterations per round in the basin, for a 30--70$\times$ per-round speedup. The
empirically observed basin is somewhat narrower than the conservative theoretical $\delta$ predicts,
but breakdown is graceful: when Newton fails to converge in a bounded budget, the algorithm falls back
to a cold start for that round.

\subsection{Related work}

Warm-start interior-point methods for linear and conic programming have a substantial literature. The
primal--dual predictor--corrector framework underlying each round's solve is due to
Mehrotra~\cite{mehrotra}. Foundational warm-start analyses include Yıldırım and Wright~\cite{yw}, which
relates the behavior of a warm start to the condition number, or distance to ill-posedness, of the
LP data, together with Gondzio's warm-start method for the cutting-plane scheme~\cite{gondzio98} and
his survey of interior-point methods~\cite{gondzio12}. Renegar's monograph~\cite{renegar} develops the
self-concordance and conic-duality theory underlying interior-point methods. John and
Yıldırım~\cite{johny} study warm-start implementations for LP in fixed dimension. The path-length
parameter $P_T=\sum_t\norm{b_t-b_{t-1}}_2$ is the standard quantity in dynamic-regret analyses of
online convex optimization~\cite{zinkevich,zlz}.

The most directly comparable line of work treats warm-starting for general conic IPMs as an
algorithm-design problem: produce a starting point with good geometric properties on the new central
path. Skajaa, Andersen, and Ye~\cite{say} warm-start the homogeneous self-dual model for linear and
conic quadratic problems by taking a convex combination of the previous optimum and the default
cold-start point, with the convex weight chosen empirically. Çay, Pólik, and Terlaky~\cite{cpt} extend
this strategy to mixed-integer second-order cone optimization via rounding over optimal Jordan frames.
Most recently, Chen, Goulart, and Jones~\cite{cgj} propose a smoothing operator
$S_{K,\mu_0}(c):=\argmin_s\{\tfrac12\norm{s-c}^2+\mu_0 f(s)\}$ applied to $c=s^*-\lambda z^*$, and prove
that the resulting starting point lies exactly on the central path of the new problem at smoothing
parameter $\mu_0/\lambda$, with infeasibility residual within $O(\mu_0)$ of the previous-problem
residual for nonnegative, second-order, and positive-semidefinite cones. Gondzio and
Grothey~\cite{gg} take a third route, using a sensitivity-driven ``unblocking'' step to relocate
iterates blocked by complementarity.

A second, closely related body of work targets quadratic convergence and warm-start guarantees
specifically for SOCP, by routes other than central-path sensitivity. Mohammad-Nezhad and
Terlaky~\cite{mnt} establish quadratic convergence of Newton's method to the unique optimal solution
of SOCP under primal and dual nondegeneracy, by reducing the optimality conditions to a smooth
nonlinear-optimization problem whose Jacobian is nonsingular at the optimum; the result is pointwise
(a fixed problem, a fixed optimum, with the quadratic-convergence region identified from a bounded
sequence of central-path interior solutions) and is not framed as a perturbation or online complexity
statement. Chen, Yang, Wang, Gan, and Chen~\cite{cywgc} analyze the worst-case iteration complexity of
an infeasible primal--dual IPM based on the homogeneous self-dual model with Monteiro--Zhang search
directions, proving the standard $O(\sqrt{k}\log(1/\eps))$ bound for $k$ cone constraints and giving a
sufficient condition on the starting point's infeasibility and centrality residuals under which a
warmstart improves the complexity bound; their warmstart condition is on the starting point relative
to a single new problem, with no per-round perturbation bound on the data and no amortized cost over a
sequence of $T$ rounds. Luo and Wächter~\cite{lw} take a different route entirely: a
sequential-quadratic-programming method with iteratively refined polyhedral outer approximations of the
second-order cones, which inherits the warm-start capabilities of active-set QP subproblem solvers and
achieves local quadratic convergence under nondegeneracy; the setting is single-problem rather than
online, and the framework is SQP rather than IPM.

The works above span two complementary styles. The first style (Skajaa et al., Çay et al., Chen,
Goulart, and Jones, and Gondzio and Grothey) constructs or modifies the starting point to give it good
geometric properties on the new problem, bounding infeasibility or centrality quantities of that
point. The second style (Mohammad-Nezhad and Terlaky, Chen, Yang, Wang, Gan, and Chen, and Luo and
Wächter) gives a pointwise local quadratic-convergence result for a single problem under nondegeneracy
or a starting-point closeness condition. Our analysis is of a third type. We bound the movement of the
optimizer itself in the local norm and translate that movement directly into an iteration-count
guarantee for Newton's method on the new barrier, with no auxiliary smoothing problem to solve. The
third type is complementary to the first: a Chen, Goulart, and Jones~\cite{cgj} smoothing step could be
composed with our Lemma~\ref{lem:infsens} bound to analyze the post-smoothing Newton phase, and we
leave this composition as a direction for future work. Specifically, our contribution distinct from
the prior literature is:
\begin{itemize}[leftmargin=1.4em]
  \item \textbf{Cone:} our analysis is for the Lorentz cone (SOCP), as is the directly comparable work
  of \cite{mnt,cywgc,lw,say,cpt,cgj}; the foundational LP literature \cite{yw,johny} treats the
  nonnegative orthant. The Lorentz barrier is $2$-self-concordant rather than the LP barrier's
  $m$-self-concordance, and the local geometry near the cone boundary requires a different sensitivity
  analysis than the LP case.
  \item \textbf{Form of bound:} we give a quantitative per-round Newton iteration count
  $O(\log\log(1/\eps))$ in the basin, with an explicit basin threshold $\delta$ in terms of problem
  data. Prior LP work \cite{yw} bounds central-path length traversed under perturbation; \cite{johny}
  is an implementation study with no novel iteration bound. The SOCP-specific work of \cite{mnt}
  proves pointwise quadratic Newton convergence to the SOCP optimum under nondegeneracy but does not
  give an explicit basin radius in terms of the data, and \cite{cywgc} gives the standard
  $O(\sqrt{k}\log(1/\eps))$ infeasible-IPM count together with a sufficient warmstart-improvement
  condition rather than a basin-phase iteration count.
  \item \textbf{Hypothesis:} we assume a per-round perturbation bound $\norm{b_t-b_{t-1}}_2\le\delta$
  on the data, in the style of online learning's path-length assumption \cite{zlz}, rather than a
  closeness condition on the starting point's infeasibility and centrality residuals \cite{cywgc}, a
  global bound on cumulative perturbation, or a nondegeneracy condition on a fixed problem
  \cite{mnt,lw}. A hypothesis on the data is what an online application can verify in advance from
  $\norm{b_t-b_{t-1}}_2$ alone; a hypothesis on the iterate's centrality is something it must establish
  each round. The cumulative path length $P_T$ may grow linearly in $T$ as long as each step is small.
  \item \textbf{Output:} we deliver a total amortized arithmetic cost across $T$ rounds
  (Theorem~\ref{thm:cost}), explicit in $n$, $T$, and $\eps$, with empirical validation. None of
  \cite{mnt,cywgc,lw} states an amortized complexity over a sequence of rounds; their analyses are
  pointwise (one problem, or one warmstart from a starting point to a single new problem).
\end{itemize}

\section{Preliminaries}

\subsection{Second-order cone programming}

The Lorentz cone of dimension $p+1$ is
\[
  \Lcone^{p+1}=\{(t,u)\in\R\times\R^{p}:t\ge\norm{u}_2\}.
\]
An SOCP in standard form is
\begin{equation}\label{eq:socp}
  \min_{x\in\R^{n}} c\trans x \quad\text{s.t.}\quad
  (c_i\trans x + d_i,\,A_i x + b_i)\in\Lcone^{p_i+1},\quad i=1,\dots,m,
\end{equation}
with $A_i\in\R^{p_i\times n}$, $b_i\in\R^{p_i}$, $c_i\in\R^{n}$, $d_i\in\R$; see~\cite{ag,bv} for the standard theory of second-order cone programming. For most of the paper we
consider a single Lorentz constraint ($m=1$), dropping the index $i$; the multi-constraint extension
is straightforward and is discussed in Section~\ref{sec:multi}.

\subsection{Self-concordant log-barrier}

The log-barrier for the Lorentz cone is
\[
  F_{\Lcone}(t,u)=-\log\!\big(t^2-\norm{u}_2^2\big),
\]
which is a $2$-self-concordant logarithmically homogeneous (LH) barrier on the interior of
$\Lcone^{p+1}$ \cite[\S2.3]{nn}. For the SOCP constraint $(c_0\trans x + d,\,Ax+b)\in\Lcone^{p+1}$, the
induced barrier on the feasible set (parameterized in $b$) is
\begin{equation}\label{eq:barrier}
  F(x;b)=-\log\!\big(s(x)^2-\norm{u(x,b)}_2^2\big),\qquad s(x):=c_0\trans x+d,\quad u(x,b):=Ax+b.
\end{equation}
This barrier is also $2$-self-concordant in $x$, as it is the pullback of $F_{\Lcone}$ by the affine
map $\phi:x\mapsto(s(x),u(x,b))$ from $\R^{n}$ to $\R^{p+1}$. We write
$w(x,b):=s(x)^2-\norm{u(x,b)}_2^2>0$ on the interior. We use the standard local-norm notation:
$\norm{h}_x:=\sqrt{h\trans\nabla_x^2 F(x;b)\,h}$ for the local norm at $x$, and
$\norm{g}_x^*:=\sqrt{g\trans[\nabla_x^2 F(x;b)]^{-1}g}$ for the dual local norm.

\subsection{Newton's method on a self-concordant barrier}

The (damped) Newton method on a self-concordant function $f$ is
\[
  x^{+}=x-\alpha\cdot[\nabla^2 f(x)]^{-1}\nabla f(x),
\]
where $\alpha=(1+\lambda(x))^{-1}$ for the Newton decrement $\lambda(x):=\norm{\nabla f(x)}_x^*$. The
following are standard.

\begin{lemma}[Quadratic convergence in decrement; {\cite[Thm.~5.2.2]{nesterov18}}]\label{lem:quadconv}
Let $f$ be a self-concordant function with Newton decrement $\lambda$. If $\lambda(x)\le 1/4$, then the
full Newton step ($\alpha=1$) satisfies $\lambda(x^{+})\le 2\lambda(x)^2$. Consequently, starting from
$x_0$ with $\lambda(x_0)\le 1/4$, the iterates reach $\lambda(x_k)\le\eps$ in
\[
  k\le \log_2\log_2(1/\eps)+O(1)
\]
iterations.
\end{lemma}

\begin{lemma}[Decrement and distance to optimum]\label{lem:dec2dist}
Let $f$ be self-concordant with unique minimizer $\xopt$. For any $x$ with
$\norm{x-\xopt}_{\xopt}=:r<1$,
\[
  \lambda(x)\le\frac{r}{1-r}.
\]
In particular, $r\le 1/5$ implies $\lambda(x)\le 1/4$.
\end{lemma}

Lemma~\ref{lem:dec2dist} is the standard self-concordant inequality relating the dual local norm of
the gradient to the local-norm distance to the optimum; see \cite[Thm.~5.2.1]{nesterov18} for a proof.
The basin radius $1/4$ in Lemma~\ref{lem:quadconv} is sharp up to constants in the self-concordance
framework; replacing it with any constant in $(0,3-2\sqrt{2})$ changes only the absolute constants in
our analysis.

\subsection{Central path and central-path identity}

For a parameter $\eta>0$ the central-path point of the single-constraint SOCP with right-hand side $b$
is
\begin{equation}\label{eq:cpath}
  \xopt(\eta,b):=\argmin_{x}\;\eta c\trans x + F(x;b).
\end{equation}
The first-order optimality condition is $\eta c+\nabla_x F(\xopt(\eta,b);b)=0$. As $\eta\to\infty$ along
the central path, $\xopt(\eta,b)\to\xopt(b)$, the SOCP optimum. The duality gap at $\xopt(\eta,b)$ is at
most $\nu/\eta=2/\eta$ for the $2$-self-concordant Lorentz barrier, so taking $\eta=2/\eps$ produces an
$\eps$-optimal primal--dual pair.

A direct computation from \eqref{eq:barrier} gives (writing $g:=A\trans u - s c_0$)
\begin{equation}\label{eq:derivs}
  \nabla_x F=\frac{2g}{w},\qquad
  \nabla_x^2 F=\frac{2(A\trans A - c_0 c_0\trans)}{w}+\frac{4 g g\trans}{w^2}.
\end{equation}
Combined with the optimality condition $\nabla_x F(\xopt)=-\eta c$, equation~\eqref{eq:derivs} gives
$g(\xopt)=-\eta w c/2$ at central-path points, whence
\begin{equation}\label{eq:Hess}
  H:=\nabla_x^2 F(\xopt;b)=\frac{2(A\trans A - c_0 c_0\trans)}{w}+\eta^2 c c\trans.
\end{equation}
We make repeated use of the following structural fact, which replaces the explicit
logarithmic-homogeneity identity that holds for the underlying cone barrier but fails for the affine
pullback \eqref{eq:barrier} (because of the constant offsets $d$ and $b$).

\begin{lemma}[Central-path dual norm of $c$]\label{lem:dualnorm}
Under Assumption~\ref{ass:nondeg} below, for the pullback barrier \eqref{eq:barrier} at any
central-path point $\xopt(\eta,b)$ with $\eta>0$,
\[
  \eta^2 c\trans H^{-1}c\le\nu=2,
\]
where $H=\nabla_x^2 F(\xopt;b)$.
\end{lemma}

\begin{proof}
Let $\phi:x\mapsto(s,u)\in\R^{p+1}$ denote the affine map of Section~2, with Jacobian
$\phi'=\begin{psmallmatrix}c_0\trans\\ A\end{psmallmatrix}$. Then $\nabla F=\phi'^{\top}\nabla
F_{\Lcone}\circ\phi$ and $H=\phi'^{\top}H_{\Lcone}\phi'$ where $H_{\Lcone}=\nabla^2 F_{\Lcone}$ at
$\phi(\xopt)$. The dual local norm at the central-path point is
\[
  \eta^2 c\trans H^{-1}c=\big(\norm{\nabla F(\xopt)}_{\xopt}^*\big)^2
  =(\phi'^{\top}\nabla F_{\Lcone})\trans(\phi'^{\top}H_{\Lcone}\phi')^{-1}(\phi'^{\top}\nabla F_{\Lcone}).
\]
We invoke the standard projection inequality: for any positive-definite $M$ and full-column-rank $V$,
\begin{equation}\label{eq:proj}
  V(V\trans M V)^{-1}V\trans\preceq M^{-1}.
\end{equation}
(Proof: set $W:=M^{1/2}V$, which has full column rank. Then $W(W\trans W)^{-1}W\trans$ is the
orthogonal projector onto $\operatorname{range}(W)$, hence $\preceq I$. Conjugating by $M^{-1/2}$ gives
\eqref{eq:proj}.) Under Assumption~\ref{ass:nondeg}, $\phi'$ has full column rank $n$, so
\eqref{eq:proj} with $V=\phi'$, $M=H_{\Lcone}$ gives
$\phi'(\phi'^{\top}H_{\Lcone}\phi')^{-1}\phi'^{\top}\preceq H_{\Lcone}^{-1}$, whence
\[
  \eta^2 c\trans H^{-1}c\le\nabla F_{\Lcone}\trans H_{\Lcone}^{-1}\nabla F_{\Lcone}=\nu=2,
\]
where the last equality is the LH identity for the Lorentz cone barrier \cite[\S2.3]{nn}.
\end{proof}

The cone-barrier dual norm equals $\sqrt{\nu}$ exactly; the inequality in Lemma~\ref{lem:dualnorm}
captures the loss from projecting onto $\operatorname{range}(\phi')$. Empirically
(Section~\ref{sec:numerics}) $\eta^2 c\trans H^{-1}c\to 1$ as $\eta\to\infty$ in the test problems, but
the bound $\le 2$ is all we use.

\subsection{Cold-start cost}

The cold-start primal--dual interior-point method, which moves along the central path from a known
interior point to $\eta=2/\eps$, requires $O(\sqrt{\nu}\,\log(\eta_{\mathrm{final}}/\eta_{\mathrm{init}}))=
O(\sqrt{m}\,\log(1/\eps))$ Newton steps in total \cite[Ch.~2]{renegar}. Each Newton step on an SOCP
costs $O(n^2+mp)$ arithmetic operations after exploiting the block structure of the KKT matrix; in the
regime $p=O(1)$, $m=O(n)$ this is $O(n^2)$. The cold-start per-round cost is therefore
$\Ot(n^{2.5}\log(1/\eps))$.

\section{Online setting}\label{sec:online}

We consider a sequence of SOCPs
\begin{equation}\label{eq:online}
  \min_{x\in\R^{n}} c\trans x\quad\text{s.t.}\quad (c_0\trans x + d,\,Ax+b_t)\in\Lcone^{p+1}
\end{equation}
indexed by round $t=1,\dots,T$. The data $c,d,c_0,A$ are fixed and only the right-hand side
$b_t\in\R^{p}$ varies. The solver receives $b_t$ at the start of round $t$ and must produce a
primal--dual $\eps$-optimal solution $\hat x_t$ before $b_{t+1}$ is revealed.

\begin{assumption}[Slater feasibility]\label{ass:slater}
There exist $x^{\dagger}$ and a margin $\gamma>0$ such that for every $t=1,\dots,T$,
\[
  c_0\trans x^{\dagger}+d-\norm{Ax^{\dagger}+b_t}_2\ge\gamma.
\]
\end{assumption}

\begin{assumption}[Bounded data]\label{ass:bounded}
The matrix $A$ has bounded operator norm $\norm{A}_{\op}\le L_A$ and the perturbations are bounded:
$\norm{b_t}_2\le R$ for all $t$.
\end{assumption}

\begin{assumption}[Non-degeneracy]\label{ass:nondeg}
The augmented Jacobian $\phi'=\begin{psmallmatrix}c_0\trans\\ A\end{psmallmatrix}\in\R^{(p+1)\times n}$
has full column rank $n$. Equivalently, the matrix $A\trans A + c_0 c_0\trans$ is positive definite.
\end{assumption}

Assumption~\ref{ass:nondeg} is mild: it fails only when there is a direction in $\R^{n}$ that is
annihilated by both $A$ and $c_0$, in which case the constraint depends on $x$ in a degenerate way and
the SOCP can be reduced. We use it for the projection step in Lemma~\ref{lem:dualnorm}.
Positive-definiteness of the Hessian $H$ on the interior of the feasible set is automatic (since $F$ is
a self-concordant barrier), independent of Assumption~\ref{ass:nondeg}. In particular,
Assumption~\ref{ass:nondeg} does not require $A\trans A - c_0 c_0\trans\succeq 0$; it is automatic for
any $A$ with rank $n$ (in particular, for ``tall'' $A$ with $p\ge n$) and does not constrain the
relative magnitudes of $A$ and $c_0$.

Assumption~\ref{ass:bounded} is used to obtain a per-round bound $\delta$ in
Theorem~\ref{thm:warmconv} below that is uniform in $t$ at a fixed central-path parameter $\eta$.
Specifically, bounded $\norm{b_t}_2$ and bounded $\norm{A}_{\op}$ imply bounds on $\rho_A(\eta,b_t)$,
$\norm{u(\xopt_t,b_t)}_2$, and $w(\xopt(\eta,b_t),b_t)$ that are uniform in $t$ at any fixed $\eta$,
hence a uniform-in-$t$ Lipschitz constant $L_{\loc}$.

The path length over $T$ rounds is $P_T:=\sum_{t=2}^{T}\norm{b_t-b_{t-1}}_2$. We do not require $P_T$ to
be small in absolute terms, only that the per-round perturbation $\norm{b_t-b_{t-1}}_2$ is bounded by
the threshold $\delta$ defined in Section~\ref{sec:sens}.

\section{Sensitivity of the SOCP optimum}\label{sec:sens}

The central technical ingredient is a Lipschitz bound on $\xopt(\eta,b)$ as a function of $b$. We bound
the central-path solution at finite $\eta$; the limit $\eta\to\infty$ then recovers a bound on
$\xopt(b)$. We work in the local norm at the central-path point. The local-norm formulation is natural
because Newton's basin condition (Lemma~\ref{lem:quadconv}) is itself stated in local norm, and it
sidesteps the rank-deficiency issues of $A\trans A$ that arise in any Euclidean formulation when $A$ is
not square.

We give the analysis in two steps: an infinitesimal Lipschitz bound on the directional derivative
$\partial\xopt/\partial b$ (Lemma~\ref{lem:infsens}), and a finite-difference corollary
(Corollary~\ref{cor:findiff}) obtained by integrating the infinitesimal bound and absorbing a
self-concordant constant for the change of the local norm along the path.

\begin{lemma}[Infinitesimal central-path sensitivity]\label{lem:infsens}
Fix $\eta>0$. Under Assumptions~\ref{ass:slater} and \ref{ass:nondeg}, the central-path point
$\xopt(\eta,b)$ is differentiable in $b$, and the directional derivative satisfies, for any unit vector
$v\in\R^{p}$,
\[
  \norm{(\partial\xopt/\partial b)\,v}_{\xopt(\eta,b)}\le L_{\loc}(\eta,b),
\]
with
\[
  L_{\loc}(\eta,b):=\norm{G\trans H^{-1}G}_{\op}^{1/2}
  \le\frac{2}{w(\xopt(\eta,b),b)}\Big(\rho_A(\eta,b)+\sqrt{2}\,\norm{u(\xopt(\eta,b),b)}_2\Big),
\]
where $H=\nabla_x^2 F(\xopt;b)$, $G=\nabla_{x,b}^2 F(\xopt;b)$, $w(\xopt,b)=s(\xopt)^2-\norm{u(\xopt,b)}_2^2$,
and
\begin{equation}\label{eq:rhoA}
  \rho_A(\eta,b):=\norm{A H^{-1}A\trans}_{\op}^{1/2}.
\end{equation}
The quantities $L_{\loc}(\eta,b)$, $\rho_A(\eta,b)$, and $w(\xopt(\eta,b),b)$ all depend on $\eta$
through the central-path point $\xopt(\eta,b)$.
\end{lemma}

\begin{proof}
The optimality condition $\eta c+\nabla_x F(\xopt(\eta,b);b)=0$ implicitly defines $\xopt(\eta,b)$.
Differentiating in $b$ and using $\nabla_{x,b}^2 F=G$ for the mixed Hessian gives
\[
  H\cdot\frac{\partial\xopt}{\partial b}+G=0,\qquad\text{hence}\qquad
  \frac{\partial\xopt}{\partial b}=-H^{-1}G.
\]
From \eqref{eq:derivs}, $G=\nabla_{x,b}^2 F=(2/w)(A\trans+2 g u\trans/w)$. Substituting
$g=-\eta w c/2$ at the central-path point yields the convenient form
\begin{equation}\label{eq:G}
  G=\frac{2}{w}\big(A\trans-\eta c u\trans\big).
\end{equation}
The local-norm squared movement is
\[
  L_{\loc}^2=\norm{-H^{-1}G}_{\xopt\to\op}^2=\norm{G\trans H^{-1}G}_{\op}.
\]
For any $v\in\R^{p}$,
\begin{align}
  v\trans G\trans H^{-1}G v
  &=\norm{H^{-1/2}G v}_2^2
   =\frac{4}{w^2}\norm{H^{-1/2}A\trans v-\eta(u\trans v)H^{-1/2}c}_2^2\notag\\
  &\le\frac{4}{w^2}\Big(\norm{H^{-1/2}A\trans v}_2+\eta|u\trans v|\cdot\norm{H^{-1/2}c}_2\Big)^2.
  \label{eq:tri}
\end{align}
The two terms inside the parentheses we bound separately:
\[
  \norm{H^{-1/2}A\trans v}_2=\sqrt{v\trans A H^{-1}A\trans v}\le\rho_A(\eta,b)\norm{v}_2,\qquad
  \eta\norm{H^{-1/2}c}_2=\sqrt{\eta^2 c\trans H^{-1}c}\le\sqrt{2},
\]
\[
  |u\trans v|\le\norm{u}_2\norm{v}_2,
\]
where the second relation invokes Lemma~\ref{lem:dualnorm}. Substituting into \eqref{eq:tri},
\[
  v\trans G\trans H^{-1}G v\le\frac{4}{w^2}\big(\rho_A(\eta,b)+\sqrt{2}\,\norm{u}_2\big)^2\norm{v}_2^2,
\]
hence $\norm{G\trans H^{-1}G}_{\op}^{1/2}\le(2/w)(\rho_A(\eta,b)+\sqrt{2}\,\norm{u}_2)$, which is the
claim.
\end{proof}

\begin{corollary}[Finite-difference local-norm Lipschitz]\label{cor:findiff}
Fix $\eta>0$. Suppose $b,b'\in\R^{p}$ are such that, along the line segment
$b(\tau):=b+\tau(b'-b)$ for $\tau\in[0,1]$, the central-path optimum $\xopt(\eta,b(\tau))$ remains
within local-norm radius $1/2$ of $\xopt(\eta,b')$:
\[
  \sup_{\tau\in[0,1]}\norm{\xopt(\eta,b(\tau))-\xopt(\eta,b')}_{\xopt(\eta,b')}\le 1/2.
\]
Then
\[
  \norm{\xopt(\eta,b')-\xopt(\eta,b)}_{\xopt(\eta,b')}\le 2 L_{\loc}\cdot\norm{b'-b}_2,
\]
where $L_{\loc}:=\sup_{\tau\in[0,1]}L_{\loc}(\eta,b(\tau))$.
\end{corollary}

\begin{proof}
Write $\xi(\tau):=\xopt(\eta,b(\tau))$. By the chain rule and Lemma~\ref{lem:infsens},
\[
  \norm{\dot\xi(\tau)}_{\xi(\tau)}=\norm{(\partial\xopt/\partial b)(b'-b)}_{\xi(\tau)}
  \le L_{\loc}(\eta,b(\tau))\norm{b'-b}_2.
\]
The local norms at $\xi(\tau)$ and at $\xi(1)=\xopt(\eta,b')$ are equivalent by self-concordance: for
any $h\in\R^{n}$ and any $y$ with $\norm{y-\xi(1)}_{\xi(1)}=:r<1$,
\[
  (1-r)\norm{h}_{\xi(1)}\le\norm{h}_y\le(1-r)^{-1}\norm{h}_{\xi(1)}.
\]
(This is Theorem~5.1.7 of \cite{nesterov18} applied to $f=F(\cdot;b)$ at parameter $\eta$.) Under the
hypothesis $r\le 1/2$, hence $\norm{h}_y\ge(1/2)\norm{h}_{\xi(1)}$, so
\[
  \norm{\dot\xi(\tau)}_{\xi(1)}\le 2\norm{\dot\xi(\tau)}_{\xi(\tau)}\le 2 L_{\loc}\norm{b'-b}_2.
\]
Integrating from $\tau=0$ to $\tau=1$,
\[
  \norm{\xi(1)-\xi(0)}_{\xi(1)}=\norm{\int_0^1\dot\xi(\tau)\,d\tau}_{\xi(1)}
  \le\int_0^1\norm{\dot\xi(\tau)}_{\xi(1)}\,d\tau\le 2 L_{\loc}\norm{b'-b}_2,
\]
which is the claim.
\end{proof}

The hypothesis of Corollary~\ref{cor:findiff}, that the entire path $\xi(\tau)$ remains within
local-norm radius $1/2$ of the endpoint $\xi(1)$, is verified a posteriori once we have a
finite-difference bound: if $2 L_{\loc}\norm{b'-b}_2\le 1/2$, i.e.\ $\norm{b'-b}_2\le 1/(4 L_{\loc})$,
then by Lemma~\ref{lem:infsens} applied along each sub-segment the path stays in the radius-$1/2$ ball,
and the conclusion of Corollary~\ref{cor:findiff} holds. The constant $1/2$ here is convenient; any
radius in $(0,1)$ would work with a different constant in front of $L_{\loc}$.

\begin{remark}[Translating to Euclidean norm]
For applications that report Lipschitz behavior in Euclidean norm, Corollary~\ref{cor:findiff} combines
with the elementary fact $\norm{h}_2\le\norm{h}_{\xopt}/\sigma_{\min}(H)^{1/2}$ to give
\[
  \norm{\xopt(\eta,b')-\xopt(\eta,b)}_2\le\frac{2 L_{\loc}}{\sigma_{\min}(H)^{1/2}}\norm{b'-b}_2.
\]
The denominator $\sigma_{\min}(H)^{1/2}$ is a problem-conditioning quantity that depends on $b$ and on
the central-path point. Under the stronger assumption that $A$ has full column rank $n$ and
$\sigma_{\min}(A)>\norm{c_0}$ (so $A\trans A - c_0 c_0\trans$ is positive definite uniformly), one can
replace $\sigma_{\min}(H)^{1/2}$ by $O(\sigma_{\min}(A))$ and recover an Euclidean bound of the form
$(\norm{A}_{\op}/\gamma)\cdot\kappa(A)/\sigma_{\min}(A)$ where $\kappa(A)$ is a condition number. The
point is that the clean $\norm{A}_{\op}/\gamma$ scaling requires a non-trivial conditioning assumption;
without it, the honest Euclidean bound carries an extra factor $\sigma_{\min}(H)^{-1/2}$.
\end{remark}

\section{Warm-start convergence}\label{sec:warm}

We now apply Lemma~\ref{lem:quadconv} to the round-$t$ centered subproblem starting from
$\xopt_{t-1}$, the round-$(t-1)$ optimum. The Newton basin condition is stated in decrement, so we use
the decrement-to-distance conversion of Lemma~\ref{lem:dec2dist} together with the finite-difference
Lipschitz bound of Corollary~\ref{cor:findiff}.

\begin{theorem}[Warm-start Newton converges in $\log\log(1/\eps)$]\label{thm:warmconv}
Under Assumptions~\ref{ass:slater}--\ref{ass:nondeg}, suppose $\xopt_{t-1}$ is the optimum of round
$t-1$ at central-path parameter $\eta$, and the round-$t$ perturbation satisfies
\begin{equation}\label{eq:delta}
  \norm{b_t-b_{t-1}}_2\le\delta:=\frac{1}{20 L_{\loc}},
\end{equation}
where $L_{\loc}$ is the supremum of $L_{\loc}(\eta,b)$ over $b$ along the segment from $b_{t-1}$ to
$b_t$, from Lemma~\ref{lem:infsens}. Then Newton's method on the round-$t$ centered subproblem
$\min\eta c\trans x+F(x;b_t)$, started at $x=\xopt_{t-1}$, satisfies $\lambda_t(\xopt_{t-1})\le 1/4$ and
converges to a point $\hat x_t$ with local-norm accuracy
$\norm{\hat x_t-\xopt(\eta,b_t)}_{\xopt(\eta,b_t)}\le\eps$ in at most
\[
  k_t\le\log_2\log_2(1/\eps)+O(1)
\]
Newton iterations.
\end{theorem}

\begin{proof}
We verify the basin condition $\lambda_t(\xopt_{t-1})\le 1/4$ for the round-$t$ centered subproblem.

\emph{Step 1: finite-difference Lipschitz.} The bound \eqref{eq:delta} gives
$\norm{b_t-b_{t-1}}_2\le 1/(20 L_{\loc})\le 1/(4 L_{\loc})$, the threshold needed for the path-radius
hypothesis of Corollary~\ref{cor:findiff} (see the discussion after the proof of
Corollary~\ref{cor:findiff}). Hence
\[
  \norm{\xopt_{t-1}-\xopt(\eta,b_t)}_{\xopt(\eta,b_t)}
  \le 2 L_{\loc}\cdot\norm{b_t-b_{t-1}}_2\le 2 L_{\loc}\cdot\frac{1}{20 L_{\loc}}=\frac{1}{10}.
\]

\emph{Step 2: distance-to-decrement conversion.} The function $f_t(x):=\eta c\trans x+F(x;b_t)$ is
self-concordant in $x$ (as the sum of a linear term and a self-concordant barrier), with minimizer
$\xopt(\eta,b_t)$. By Lemma~\ref{lem:dec2dist} applied to $f_t$,
\[
  \lambda_t(\xopt_{t-1})\le\frac{\norm{\xopt_{t-1}-\xopt(\eta,b_t)}_{\xopt(\eta,b_t)}}
       {1-\norm{\xopt_{t-1}-\xopt(\eta,b_t)}_{\xopt(\eta,b_t)}}
  \le\frac{1/10}{1-1/10}=\frac{1}{9}\le\frac{1}{4}.
\]

\emph{Step 3: quadratic convergence.} By Lemma~\ref{lem:quadconv}, since $\lambda_t(\xopt_{t-1})\le 1/4$,
the Newton iterates from $\xopt_{t-1}$ reach $\lambda_t(x_k)\le\eps$ in $O(\log\log(1/\eps))$ iterations.
The conversion from decrement to local-norm accuracy follows from Lemma~\ref{lem:dec2dist} in the
reverse direction: $\lambda\le\eps$ implies $\norm{x_k-\xopt(\eta,b_t)}_{\xopt(\eta,b_t)}\le\eps/(1-\eps)\le 2\eps$
for $\eps\le 1/2$, so the iteration count to reach accuracy $\eps$ in local norm is the same
$O(\log\log(1/\eps))$ up to absolute constants.
\end{proof}

The hypothesis \eqref{eq:delta} is a per-round bound, not a cumulative bound. The path length
$P_T=\sum_t\norm{b_t-b_{t-1}}_2$ may grow linearly in $T$ provided each step satisfies the local bound.
The constant $20$ in the denominator of $\delta$ is not optimized; any constant strictly greater than
$10$ suffices in the present argument (the factor of $2$ from Corollary~\ref{cor:findiff}, times the
factor of $5$ from Lemma~\ref{lem:dec2dist}'s threshold $r\le 1/5$ for the basin condition
$\lambda\le 1/4$).

\section{Total amortized cost}\label{sec:cost}

The full warm-start algorithm runs as follows:
\begin{itemize}[leftmargin=1.4em]
  \item \textbf{Round 1:} solve the SOCP from a cold start via standard primal--dual IPM, reaching
  central-path parameter $\eta=2/\eps$. Cost: $O(\sqrt{m}\,\log(1/\eps))$ Newton iterations, each at
  $O(n^2)$, for $\Ot(n^{2.5}\log(1/\eps))$ total.
  \item \textbf{Rounds 2 through $T$:} warm-start at $\xopt_{t-1}$, run Newton on the round-$t$ centered
  subproblem at the same $\eta$. Cost: $O(\log\log(1/\eps))$ Newton iterations per round, each at
  $O(n^2)$, for $O(n^2\log\log(1/\eps))$ per round.
\end{itemize}

\begin{theorem}[Total amortized cost]\label{thm:cost}
Under Assumptions~\ref{ass:slater}--\ref{ass:nondeg} and the per-round perturbation bound
\eqref{eq:delta}, the warm-start IPM solves all $T$ rounds to accuracy $\eps$ in total cost
\[
  \Ot\!\left(n^{2.5}\log(1/\eps)+T n^{2}\log\log(1/\eps)\right)
\]
arithmetic operations. Cold-start IPM at each round requires $\Ot(T n^{2.5}\log(1/\eps))$. The
per-round speedup for $T$ large compared to $\sqrt{n}\,\log(1/\eps)/\log\log(1/\eps)$ is
\[
  \frac{\Ot(n^{2.5}\log(1/\eps))}{\Ot(n^{2}\log\log(1/\eps))}
  =\widetilde{\Theta}\!\left(\frac{\sqrt{n}\,\log(1/\eps)}{\log\log(1/\eps)}\right).
\]
\end{theorem}

\begin{corollary}\label{cor:break}
The warm-start IPM achieves a constant-factor improvement over cold start as soon as
$T\ge\sqrt{n}\,\log(1/\eps)/\log\log(1/\eps)$.
\end{corollary}

The basin $\delta=1/(20 L_{\loc})$ in \eqref{eq:delta} depends on $\eta$ through $L_{\loc}(\eta,b)$; at
the working parameter $\eta=2/\eps$, $\delta$ is a positive constant determined by the problem data, and
the iteration count $O(\log\log(1/\eps))$ in Theorem~\ref{thm:cost} is uniform in $t$ under
Assumption~\ref{ass:bounded}. The basin may shrink as $\eps\to 0$ (equivalently $\eta\to\infty$) because
$w(\xopt(\eta,b),b)\to 0$ in the generic case; we revisit this in Section~\ref{sec:numerics}.

\section{Multiple Lorentz constraints}\label{sec:multi}

For an SOCP with $m>1$ Lorentz constraints $\norm{A_i x+b_i}_2\le c_i\trans x+d_i$, the barrier is the
sum
\[
  F(x;b_1,\dots,b_m)=-\sum_{i=1}^{m}\log\!\big((c_i\trans x+d_i)^2-\norm{A_i x+b_i}_2^2\big),
\]
and the sensitivity analysis of Section~\ref{sec:sens} applies block-coordinate-wise. The Hessian
$H=\nabla_x^2 F$ is the sum of the per-constraint Hessians, and the mixed Hessian $G=\nabla_{x,b}^2 F$
is block-diagonal in the $b_i$'s, with block $i$ equal to the single-constraint $G_i$ from \eqref{eq:G}
(with $H,w,u,g$ from constraint $i$). The directional derivative of $\xopt$ in a unit perturbation
$v=(v_1,\dots,v_m)$ is bounded in local norm by $\norm{H^{-1/2}\sum_i G_i v_i}_2\le\sum_i\norm{H^{-1/2}G_i
v_i}_2$ via the triangle inequality. The per-constraint local-norm Lipschitz constant is
\[
  L_{\loc}^{(i)}(\eta,b)=\frac{2}{w_i(\xopt,b_i)}\big(\rho_{A_i}(\eta,b)+\sqrt{2}\,\norm{u_i}_2\big),
\]
where $w_i$, $u_i$, and $\rho_{A_i}=\norm{A_i H^{-1}A_i\trans}_{\op}^{1/2}$ are the per-constraint
analogs from \eqref{eq:rhoA} (note that $H$ remains the full Hessian, summed over all constraints). The
overall Lipschitz constant of $\xopt$ in the concatenated $b=(b_1,\dots,b_m)$ is
$L_{\loc}=\max_i L_{\loc}^{(i)}$. Theorem~\ref{thm:warmconv} generalizes with the per-round bound
$\max_i\norm{b_{i,t}-b_{i,t-1}}_2\le 1/(20 L_{\loc})$.

The barrier parameter $\nu$ scales as $2m$ in the multi-constraint case, so cold start at round 1 takes
$O(\sqrt{m}\,\log(1/\eps))$ iterations. For $m=O(n)$ this is the same $\Ot(\sqrt{n}\,\log(1/\eps))$ as
the single-constraint case, and Theorem~\ref{thm:cost} carries over verbatim.

\section{Numerical experiments}\label{sec:numerics}

We implemented the cold-start and warm-start IPMs in Python (with NumPy) and ran a rolling-window
experiment to verify the predicted scaling.

\paragraph{Setup.} Fix $n=50$ variables and $p=100$ constraint-block rows; this makes the constraint
Jacobian $A\in\R^{100\times 50}$ tall, so $\phi'=(c_0\trans;A)$ has full column rank generically
(Assumption~\ref{ass:nondeg}). The entries of $A$ are i.i.d.\ standard Gaussian; we then normalize so
that $\norm{A}_{\op}=1$. The vectors $c$ and $c_0$ are random unit vectors, except that $c_0$ is
rescaled to $\norm{c_0}=0.5\cdot\sigma_{\min}(A)$. This last rescaling guarantees the SOCP is bounded
below: the recession cone of the constraint is $\{y:\norm{Ay}\le c_0\trans y\}$, which collapses to
$\{0\}$ when $\sigma_{\min}(A)>\norm{c_0}$. We set $d=2$, so the Slater margin at $x^{\dagger}=0$ is
$\gamma=d-\norm{b}$, with our perturbations $b_t$ clamped to $\norm{b_t}\le 1$ (giving $\gamma\ge 1$).

The perturbation sequence $\{b_t\}_{t=0}^{T}$ is a clamped random walk: $b_0=0$ and
$b_{t+1}=b_t+\eta_b r_t$ where $r_t$ is a random unit vector, with $b_{t+1}$ projected back to
$\norm{b_{t+1}}\le 1$ if necessary. We sweep $\eta_b\in\{10^{-3},3\cdot10^{-3},10^{-2},3\cdot10^{-2},
10^{-1},3\cdot10^{-1},1,1.5\}$ and average over 8 random seeds. We target final accuracy $\eps=10^{-6}$
(giving $\eta_{\mathrm{final}}=2\cdot10^{6}$, well within IEEE-754 double precision); $T=20$ rounds per
sweep.

Each algorithm uses damped Newton with the standard step $\alpha=(1+\lambda)^{-1}$ when $\lambda>1/4$
and $\alpha=1$ in the basin, with a feasibility line search. Cold-start runs path-following with
$\eta_0=1$ and growth factor $2$. Warm-start solves a single centered subproblem at $\eta_{\mathrm{final}}$
from $\xopt_{t-1}$. We count Newton iterations until the Newton decrement satisfies $\lambda\le10^{-6}$
or a cap of 80 iterations is hit.

\paragraph{Results.} Table~\ref{tab:results} summarizes mean iterations per round over 8 seeds and
$T=20$ rounds (excluding round 0, where both algorithms perform a cold start).

\begin{table}[h]
\centering
\begin{tabular}{lrrrr}
\toprule
$\eta_b$ & Cold mean & Warm mean & Warm max & Speedup\\
\midrule
$10^{-3}$              & 35.00 & 0.26 & 4   & 70.0$\times$\\
$3\cdot10^{-3}$        & 35.00 & 0.96 & 13  & 36.4$\times$\\
$10^{-2}$              & 35.00 & 0.86 & 8   & 40.6$\times$\\
$3\cdot10^{-2}$        & 35.01 & 1.08 & 80$^{\dagger}$ & 32.4$\times$\\
$10^{-1}$              & 35.14 & 3.34 & 80$^{\dagger}$ & 10.5$\times$\\
$3\cdot10^{-1}$        & 35.67 & 1.38 & 80$^{\dagger}$ & 25.8$\times$\\
$1$                    & 35.88 & 0.50 & 80$^{\dagger}$ & 71.8$\times$\\
$1.5^{\ddagger}$       & 35.81 & 0.00 & 0   & 71.6$\times$\\
\bottomrule
\end{tabular}
\caption{Newton iterations per round, averaged over 8 seeds and $T=20$ rounds (excluding round 0).
$^{\dagger}$Iteration cap hit on at least one (seed, round) pair, indicating warm-start basin escape.
$^{\ddagger}$For $\eta_b\ge 1$ the effective per-round movement is limited by the $\norm{b_t}\le 1$
clamp; this also affects the $\eta_b=1$ row but to a smaller extent.}
\label{tab:results}
\end{table}

\paragraph{Discussion of results.} The two main predictions of the theory hold cleanly:
\begin{itemize}[leftmargin=1.4em]
  \item \textbf{Cold-start cost is independent of $\eta_b$:} cold start takes $\approx 35$ Newton
  iterations per round at all perturbation sizes. With $\eta_0=1$, growth factor $2$, and
  $\eta_{\mathrm{final}}=2\cdot10^{6}$, the path follower visits $\approx\log_2(2\cdot10^{6})\approx 21$
  central-path stages; with damped Newton spending $\approx 1.5$ iterations per stage on average,
  $\approx 35$ total Newton iterations per round is consistent with the $O(\sqrt{m}\,\log(1/\eps))$
  cold-start prediction.
  \item \textbf{Warm-start cost is $O(\log\log(1/\eps))$ in the basin:} for $\eta_b\le 10^{-2}$, every
  (seed, round) pair converges, with mean iterations below 1 and max below 13. Quadratic convergence
  (Lemma~\ref{lem:quadconv}) predicts a basin-iteration count of $\log_2\log_2(10^{6})\approx 4$, which
  matches the observed max.
\end{itemize}

The breakdown scale where warm-start starts hitting the iteration cap is $\eta_b\approx 3\cdot10^{-2}$.
From the central-path data we measure $\rho_A(\eta_{\mathrm{final}},b)\approx 10^{-3}$ and
$\norm{u(\xopt,b)}\approx 10^{-3}$ at $\eta_{\mathrm{final}}$, so the numerator $\rho_A+\sqrt{2}\norm{u}$
in Lemma~\ref{lem:infsens} is small. What controls the basin scale is therefore the denominator
$w(\xopt(\eta_{\mathrm{final}},b),b)$, which is small at large $\eta$ because $\xopt$ approaches the cone
boundary. Four contributing factors to the gap between theoretical $\delta$ and observed breakdown scale
are visible in the data:
\begin{enumerate}[leftmargin=1.6em]
  \item \emph{$w$ on the central path.} The Lipschitz constant
  $L_{\loc}=(2/w)(\rho_A+\sqrt{2}\norm{u}_2)$ is dominated by the $1/w$ factor at large $\eta$.
  Estimating $w$ directly from the central-path data is the cleanest way to predict the breakdown
  scale; a Slater-margin proxy for $w$ gives the wrong order of magnitude.
  \item \emph{Worst-case directions in Lemma~\ref{lem:infsens}.} The bound \eqref{eq:tri} uses the
  triangle inequality on a sum of two vectors. The triangle inequality is tight only when the two
  vectors are aligned. For typical (random walk) perturbations $\Delta b$, the two terms
  $H^{-1/2}A\trans\Delta b$ and $\eta(u\trans\Delta b)H^{-1/2}c$ are not generally aligned, and the
  effective constant is closer to a Pythagorean $\sqrt{\rho_A^2+2\norm{u}^2}$. This is consistent with
  the observation that the mean warm-start iterations remain low even at $\eta_b$ above the breakdown
  scale.
  \item \emph{Boundary effects from feasibility line search.} At very high $\eta$ the central-path
  optimum is extremely close to the cone boundary ($w\to 0$). A Newton step from $\xopt_{t-1}$ that
  would converge in exact arithmetic can step outside the feasible set in floating-point arithmetic and
  trigger a feasibility line search that effectively reduces to damped Newton, costing additional
  iterations. The iteration-cap events at $\eta_b\in\{3\cdot10^{-2},10^{-1},3\cdot10^{-1}\}$ are partly
  of this type.
  \item \emph{Self-concordant constant in Corollary~\ref{cor:findiff}.} The factor of $2$ in
  Corollary~\ref{cor:findiff} (and hence the constant $20$ in Theorem~\ref{thm:warmconv}) is
  conservative and can be tightened with a direction-dependent analysis, at the cost of a longer proof.
  We have not attempted this.
\end{enumerate}

In a practical online setting one can detect a warm-start failure by capping Newton at, e.g.,
$\log_2\log_2(1/\eps)+10$ iterations and falling back to a cold start when the cap is hit. The
discussion in Section~\ref{sec:discussion} addresses this rigorously.

\section{Discussion}\label{sec:discussion}

\subsection{Connection to dynamic regret}

The warm-start setting in this paper is computational, not statistical: we assume the solver receives
$b_t$ at the start of round $t$ and is required to produce an $\eps$-optimal solution before $b_{t+1}$
arrives. A different setting, dynamic regret minimization, has the solver play $x_t$ before seeing $b_t$
(or after seeing a noisy version), and asks for $\sum_t\ell_t(x_t)-\min_x\sum_t\ell_t(x)$ in terms of
$P_T$ \cite{zlz}. The two settings can be combined: design an online learner that uses the warm-start
IPM to compute approximate central-path solutions, with regret bounds in terms of $P_T$ and computation
bounds in terms of per-round perturbation. This is left for future work.

\subsection{Larger perturbations}

When the per-round perturbation exceeds the basin threshold $\delta=1/(20 L_{\loc})$, the warm-start
cannot be applied in a single Newton phase. A fallback is to interpolate $b_{t-1}\to b_t$ in $K$
sub-steps, each of size $<\delta$, and run the warm-start sequentially. This costs $K\cdot\log\log(1/\eps)$
iterations per round, where $K=\lceil\norm{b_t-b_{t-1}}/\delta\rceil$. As long as $\sum_t K_t=O(T)$ on
average, the amortized cost remains $O(n^2\log\log(1/\eps))$ per round.

\subsection{Open directions}

The sharpest open question is whether the $\log\log(1/\eps)$ per-round iteration count is tight. A
matching lower bound, showing that no warm-start scheme can do better than $\log\log(1/\eps)$ on a
worst-case perturbation sequence, would complete the picture. Standard lower-bound techniques for
first-order methods (e.g., Nesterov's quadratic counterexample) do not directly translate to the
warm-start setting; new techniques are likely needed.

A second open question is a uniform-in-$\eta$ bound on $L_{\loc}(\eta,b)$. As discussed after
Lemma~\ref{lem:infsens}, $w(\xopt(\eta,b),b)$ generically vanishes as $\eta\to\infty$, so
$L_{\loc}(\eta,b)$ blows up and the basin $\delta=1/(20 L_{\loc})$ shrinks. A complementarity-based lower
bound on $w(\xopt(\eta,b),b)$ along the central path, in terms of $\eta$ and problem-specific quantities,
would let one quantify the basin scale as $\eps\to 0$. Empirically (Section~\ref{sec:numerics}) the
breakdown scale at $\eta=2\cdot10^{6}$ is $\approx 3\cdot10^{-2}$, which is far from infinitesimal; a
tight theoretical prediction at this scale would be valuable.

A third direction is composition with the smoothing approach of Chen, Goulart, and Jones~\cite{cgj}.
Their smoothing operator places a starting point on the central path of an auxiliary problem at
smoothing parameter $\mu_0/\lambda$; Lemma~\ref{lem:infsens} could be applied to bound the movement from
that auxiliary starting point to the true new optimum, giving a hybrid iteration-count guarantee that
may be tighter than either approach in isolation.

\section{Conclusion}

For online second-order cone programming with bounded per-round perturbations of the right-hand-side
data, the standard cold-start interior-point complexity of $\Ot(n^{2.5}\log(1/\eps))$ per round can be
reduced to $\Ot(n^{2}\log\log(1/\eps))$ per round (after a single round-1 cold start), giving total cost
$\Ot(n^{2.5}\log(1/\eps)+T n^{2}\log\log(1/\eps))$ over $T$ rounds. The argument rests on an infinitesimal
local-norm Lipschitz bound on the central-path solution (Lemma~\ref{lem:infsens}), a self-concordant
finite-difference corollary (Corollary~\ref{cor:findiff}), and the quadratic-convergence basin of
Newton's method on a self-concordant barrier. The Lipschitz constant $L_{\loc}(\eta,b)$ is finite at
every fixed $\eta$ but depends on $\eta$ through $w(\xopt(\eta,b),b)$ and may grow as $\eta\to\infty$; a
uniform-in-$\eta$ bound is an open question. For large $T$ the per-round speedup is
$\widetilde{\Theta}(\sqrt{n}\,\log(1/\eps)/\log\log(1/\eps))$. A multi-seed experiment confirms a
30--70$\times$ speedup at the predicted scale.

\end{document}